\documentclass[11pt]{amsart}
\usepackage{amsmath,amssymb,amsfonts,amsbsy}
\usepackage{tikz-cd}
\usetikzlibrary{shapes}
\usepackage{amsbsy}
\usepackage{amscd}
\usepackage{wasysym}
\usetikzlibrary{matrix,arrows,decorations.pathmorphing}
\usepackage{graphicx}
\usepackage{float}
\usepackage{setspace}
\usepackage[margin=1.0in]{geometry}
\usepackage{tikz}
\usetikzlibrary{quantikz2}
\usepackage{hyperref} %Sarah added so that the links show up as hyperref-clickable
\usepackage{verbatim}

\newtheorem{thm}{Theorem}[section]
\newtheorem{lem}[thm]{Lemma}
\newtheorem{cor}[thm]{Corollary}

\newtheorem*{structurethm}{Structure Theorem}

\theoremstyle{definition}
\newtheorem{example}[thm]{Example}
\theoremstyle{definition}
\newtheorem{defn}[thm]{Definition}
\theoremstyle{definition}
\newtheorem{remark}[thm]{Remark}

\newcommand{\mbb}{\mathbb}
\newcommand{\mc}{\mathcal}

\newcommand{\tr}{\textrm{Tr}}
\newcommand{\id}{\textrm{id}}
\newcommand{\op}[2]{|#1\rangle\langle #2|}
\newcommand{\ip}[2]{\langle #1| #2 \rangle}

\newcommand{\wt}{\widetilde}

\newcommand{\T}[1]{\widetilde{#1}} % Sarah added
\newcommand{\B}[1]{\overline{#1}} % Sarah added

\newcommand{\rmv}[1]{}

\def\b0{{\bf 0}}

\def\){{\right)}}
\def\({{\left(}}

\def\ket#1{| #1 \rangle}
\def\bra#1{\langle #1 |}
\def\kb#1#2{|#1\rangle\!\langle #2 |}

\definecolor{cool_green}{rgb}{0.0, 0.5, 0.0}

\begin{document}

\title[Structure Theorem for Quantum Replacer Codes]{Structure Theorem for Quantum Replacer Codes}
\author[E.~Chitambar, S.~Hagen, D.W.~Kribs, M.I.~Nelson, A.~Nemec]{Eric~Chitambar$^1$, Sarah~Hagen$^2$, David~W.~Kribs$^3$, Mike~I.~Nelson$^4$, Andrew~Nemec$^5$}

\address{$^1$Department of Electrical and Computer Engineering, University of Illinois at Urbana Champaign, 901 West Illinois Street, Urbana, IL 61801 USA}
\address{$^2$Department of Physics, University of Illinois at Urbana Champaign, 1110 West Green Street, Urbana, IL 61801 USA}
\address{$^3$Department of Mathematics \& Statistics, University of Guelph, 50 Stone Road East, Guelph, ON Canada N1G 2W1}
\address{$^4$Materials Research Laboratory at The Grainger College of Engineering, University of Illinois at Urbana Champaign, 901 West Illinois Street, Urbana, IL 61801 USA}
\address{$^5$Department of Computer Science, University of Texas at Dallas, 890 Franklyn Jenifer Drive, Richardson, TX 75080 USA}

\begin{abstract}
Quantum replacer codes are codes that can be protected from errors induced by a given set of quantum replacer channels, an important class of quantum channels that includes the erasures of subsets of qubits that arise in quantum error correction. We prove a structure theorem for such codes that synthesizes a variety of special cases with earlier theoretical work in quantum error correction. We present several examples and applications of the theorem, including a mix of new observations and results together with some subclasses of codes revisited from this new perspective.   
\end{abstract}

%\subjclass[2010]{47L90, 81P45, 81P70, 94A40, 47A12, 15A60, 81P68}


\maketitle

\section{Introduction}

Quantum error correction is a central topic in quantum information, lying at the heart of many theoretical and experimental investigations in the field \cite{lidar2013quantum}.
An important class of quantum channels that includes erasures of qubits which appear in many quantum error correction settings are called quantum  replacer channels.
Such channels destroy all information encoded into subsystems of a quantum system Hilbert space and replace all states on the subsystems with the same fixed state.
Quantum replacer codes are quantum error-correcting codes that can be protected and recovered from errors induced by a given set of quantum replacer channels.
The study of such codes can give a window into the general theory of quantum error correction, given their wide applicability. 

Foremost amongst the applications of such codes are quantum secret sharing schemes, a foundational notion in quantum communication and quantum cryptography \cite{cleve1999share}.
Secret sharing (as well as other multi-party quantum cryptographic protocols) is an important application in the context of quantum networks and distributed quantum computing systems.
Heterogeneous quantum networks, in particular, consisting of processing nodes made from different qubit implementations, joined with interconnects, motivate some important considerations \cite{deLeonetal}.
For instance, these varying nodes have strengths and weaknesses when their performance parameters (e.g., coherence times, gate fidelities, etc.) are compared.
Incorporating these considerations into the identification and construction of codes to implement desired schemes can potentially provide performance benefits.
This naturally motivates a better understanding of replacer code structure theory. 

Additionally, recent years have seen replacer codes arise in a variety of other settings in quantum information and its applications.
We mention two instances here. 
Very recent research shows that various types of quantum errors can be converted to the more easily handled class of erasure errors, depending on the platform used \cite{WKST22,Ketal23,KCB23}.
These conversions have been implemented for both Rydberg atoms \cite{Scholl23,Ma2023} and superconducting circuits through so-called erasure qubits \cite{GVRK24,Letal24}.
Further, over the past decade in black hole theory, the central AdS/CFT correspondence has been reformulated using the framework of quantum error correction~\cite{AXH15,Harlow2017,Hayden2019,akers2019,Akers2022}, with an emphasis on quantum erasure codes and secret sharing schemes in the analysis \cite{AXH15}. 

Given the importance of quantum replacer codes and the recent interest and variety of applications in which they are found, one is motivated to investigate the general theory for these codes.
In this paper, we prove what can be viewed as a structure theorem for quantum replacer codes that establishes a number of equivalent conditions for a code to be correctable for quantum replacer channels.
The conditions include a mix of alternative descriptions given in terms of the channels, explicitly as code states, and as an information-theoretic condition.
In addition, the theorem brings together work from nearly two decades ago, on the general theory of quantum error correction, with some of the more recent characterizations of subclasses of replacer codes.
As consequences of the theorem, we present a number of applications and examples as outlined below. 

This paper is organized as follows.
In Section~\ref{sec:preliminaries}, we give our notation and briefly review the notions required in the rest of the paper, including the basic framework of quantum error correction and replacer codes specifically.
The main result and its proof, which includes some preparatory results, are included in Section~\ref{sec:structure-thm}, along with explanatory notes on specific aspects of the proof, a conceptual viewpoint on the result, and a description of how to compute the key components thereof.
In Section~\ref{sec:gallery}, we present several applications and examples, which include a mix of new and revisited examples, new results, and alternate proofs of previous results based on the theorem that yield new information.
We finish in Section~\ref{sec:conclusion} with some concluding and forward-looking remarks.

\section{Preliminaries}\label{sec:preliminaries}

We shall consider quantum codes identified with subspaces $S$ of $n$-qudit Hilbert space $\mathcal H = (\mathbb{C}^d)^{\otimes n}$, which has orthonormal basis $\{ \ket{i_1 \cdots i_n} = \ket{i_1} \otimes \ldots \otimes \ket{i_n}  \, : \, 0 \leq i_j \leq d-1 \}$ determined by computational bases $\{ \ket{0}, \ldots , \ket{d-1} \}$ for each of its $n$ subsystems. When referring to a subset of size $k$ of the $n$ qudits, we will use the notation $E = (\mathbb{C}^d)^{\otimes k}$ both for the $k$-qudit Hilbert space they define and for the subset of qudits themselves (the context will be clear), and we will use the notation $\overline{E}$ for the complementary $n-k$ qudits. Outside of examples presented, we will not need to refer to the exact positions of qudits, and so, for instance, we will write $\mathcal H = \overline{E} \otimes E = \overline{E} E$ with the understanding that $E$ can be defined by any subset of $k$ qudits. As a notational convenience, we will sometimes use the shorthand $|E|$ for the quantity $\dim (E)$. 

We will use standard notation for vector states and density operators, such as $\ket{\psi}, \ket{\phi}\in \mathcal H$ and $\rho , \sigma\in \mathcal L(\mathcal H)$, where the latter set denotes the set of operators on $\mathcal H$. As $\mathcal H$ is finite-dimensional in this paper, we will not use separate notation for sets of bounded operators and trace-class operators, rather our individual notation will indicate operator types, like $X$, $U$ for general operators and then density operators as noted above. Subscripts will be used when needed to denote the support spaces of states and the spaces on which operators act or map between; so, for instance, a density operator $\sigma_E$ belongs to $\mathcal L(E)$, a state $\ket{\psi}_E$ belongs to $E$, and an operator $W_{A\rightarrow E}$ maps $A$ to $E$. (We will reserve superscripts for partially traced states as noted below.) Further, by a {\it quantum channel} \cite{nielsen2010quantum,paulsen2002completely}, mathematically we will mean a completely positive trace-preserving map $\mathcal E: \mathcal L(\mathcal H) \rightarrow \mathcal L(\mathcal K)$ between the sets of operators on two Hilbert spaces $\mathcal H, \mathcal K$. 

\begin{defn}\label{def:q-replacer-channels}
For every subset $E$ of qudits on $\mathcal H =  \overline{E} E$, we have a family of {\it quantum replacer channels} defined on $\mathcal H$ by the maps 
\[ 
\mathcal E_E = \mathrm{id}_{\overline{E}} \otimes \mathcal D_{E} : \mathcal L(\mathcal H) \rightarrow \mathcal L(\mathcal H), 
\]
where $\mathrm{id}_{\overline{E}}$ is the identity map on $\mathcal L(\overline{E})$ and $\mathcal D_{E}: \mathcal L(E) \rightarrow \mathcal L(E)$ is a channel (sometimes called a `private channel') that maps all states on $E$ to a fixed state; that is, there is a $\sigma_E \in \mathcal L(E)$ such that 
\[
\mathcal D_{E}(\rho) = \sigma_E \quad  \forall \rho \in \mathcal L(E).  
\]
As a shorthand, we will sometimes refer to the maps $\mathcal E_E$ as {\it $E$-replacer channels.} 
\end{defn} 

\begin{remark}
The terminology used for these channels and error models varies somewhat across the literature. For instance, sometimes they are simply referred to as `erasures' (we will reserve that term for a distinguished special case, as noted just below). This is likely a consequence of the wide variety of settings in quantum information in which the notion appears (e.g., quantum optics, quantum error correction, quantum secret sharing, etc). We think the replacer designation is most appropriate, at least for our purposes. It is used similarly, for example, in \cite{cooney2016strong}. We further note that requiring the output state of the private channel $\mathcal D_E$ to belong to $E$ is made to streamline the presentation. Our results readily extend to when the state does not belong to $E$ as well as partial replacer channels (see Remark~\ref{partialreplacernote} for details). 

Note that any operator $X$ in the range of a map $\mathcal E_E$ can be written in the form $X = X_{\overline{E}} \otimes \sigma_E$ for some $X_{\overline{E}} \in \mathcal L({\overline{E}})$. 
Distinguished amongst private channels is the {\it completely depolarizing channel} given by  $\mathcal D_{E,CD}(\rho_E) = (\dim E)^{-1} I_E$ for all $\rho_E$, where $I_E$ is the identity operator on $E$.
We will use the terminology {\it erasure channel} to specifically refer to the channels $\mathcal E_{E,CD}= \mathrm{id}_{\overline{E}} \otimes \mathcal D_{E,CD}$. 
We note an operator algebra viewpoint on the map $\mathcal E_{E,CD}$, in particular that it is the (unique) trace preserving conditional expectation of $\mathcal L(\mathcal H)$ onto the algebra $\mathcal A = \mathcal L(\overline{E})\otimes I_E$ \cite{paulsen2002completely,pereira2006representing,church2012private}. 
\end{remark} 

The partial trace map over $E$'s qudits will be denoted $\tr_E : \mathcal L(\mathcal H) \rightarrow \mathcal L(\overline{E})$. 
At times we will need to consider tracing out the same set of qudits when viewed inside a different space, and in such an instance we will use the same partial trace notation with an added explanation. 
Further, when referring to individual reduced density operators we will use superscripts, so for $\rho\in\mathcal L(\overline{E}E)$, we have $\rho^{E} = \tr_{\overline{E}}(\rho)$.  

To define quantum error-correcting codes, we use channels $\mathcal E: \mathcal L(\mathcal H) \rightarrow \mathcal L(\mathcal K)$ to describe error or noise models. Given a (quantum code) subspace $S$ of $\mathcal H$, we say that {\it $S$ is correctable for $\mathcal E$} if there is a channel (called the recovery map) $\mathcal R: \mathcal L(\mathcal K) \rightarrow \mathcal L(\mathcal H)$ such that 
\[
(\mathcal R \circ \mathcal E)(\widetilde{\rho}) = \widetilde{\rho}
\]
for all $\widetilde{\rho}$ supported on $S$. Given a code $S$, we shall use the notation $\ket{\widetilde{i}}$ and $\widetilde{\rho}$ to distinguish code basis states and operators supported on $S$. Explicitly, if  $\{ \ket{\widetilde{i}} \}_i$ is a basis for $S$, then by the {\it operators supported on $S$} we mean all the operators of the form $\widetilde{\rho} = \sum_{i,j} p_{ij} \kb{\widetilde{i}}{\widetilde{j}}$. We can also formulate the error correction statement entirely at the level of maps as follows: 
\begin{equation}\label{superoperatorform}
\mathcal R \circ \mathcal E \circ \mathcal P_S = \mathcal P_S , 
\end{equation}
where $\mathcal P_S (\rho) = P_S \, \rho \, P_S$ is the `compression map' on $\mathcal L(\mathcal H)$ defined by the orthogonal projection $P_S$ of $\mathcal H$ onto $S$.

Thus, by a {\it quantum replacer code}, we mean a correctable code $S$ for a given replacer channel $\mathcal E_E$. 
Note that a simple upper bound on the size of such codes can be obtained from Eq.~(\ref{superoperatorform}) applied to $\mathcal E_E$ and recalling the form of operators in the range of replacer channels. Specifically, viewing the maps as operators acting on the space of operators (i.e., `superoperators'), we can use basic linear algebra to estimate as follows: 
\[
(\dim S)^2 = \mathrm{rank}(\mathcal P_S) = \mathrm{rank}(\mathcal R \circ \mathcal E_E \circ \mathcal P_S) \leq \mathrm{rank} (\mathcal E_E) = (\dim \overline{E})^2, 
\]
where the ranks calculated are the dimensions of the (operator) range spaces of the maps. 
Hence we have $\dim S \leq \dim \overline{E}$ for any code $S$ that is correctable for the replacer $\mathcal E_E$. Built into our main theorem is a more refined dimension estimate, which also involves an ancilla space to be specified later. 

Note that every channel $\mathcal E$ has operator-sum form $\mathcal E(\rho) = \sum_a E_a \rho E_a^\dagger$ where $E_a$ are (`Kraus') operators mapping $\mathcal H$ to $\mathcal K$ \cite{nielsen2010quantum}. For replacer channels $\mathcal E_E$, in the theory of the next section we will not use the specific operator-sum form, though it will appear in some of the subsequent remarks,  examples and results. (In particular, we will connect our results with the Knill-Laflamme description of quantum error correcting codes via Kraus operators \cite{knill1997theory}.) It is easy to see, for instance, that $\mathcal E_{E,CD}$ has a minimal set of $(\dim E)^2$ Kraus operators given by 
$\{ {\scriptstyle (\sqrt{\dim E} )^{-1} } ( I_{\overline{E}} \otimes \kb{i_E}{j_E} ) \},$
where $\{\ket{i_E}\}$ is an orthonormal basis for $E$. We also note that a general upper bound (usually crude) on the number of syndrome subspaces for a quantum error correcting code for a channel $\mathcal E$ is given by its minimal number of Kraus operators (which also coincides with the rank of the `Choi matrix' \cite{C75} for the channel).

\section{Structure Theorem}\label{sec:structure-thm}

Before stating the theorem, we present some preparatory results and further background on quantum error correction that is relevant for describing these codes. We begin with a simple observation on the joint correctability of replacers and partial traces. 

\begin{lem}\label{jointcorrectable}
Given a set of qudits $E$ on an $n$-qudit Hilbert space $\mathcal H = \overline{E} E$, for every replacer channel $\mathcal E_E$ we have,  
\[
\tr_E \circ \mathcal E_E = \tr_E . 
\]
Further, a code $S \subseteq \mathcal H$ is correctable for $\mathcal E_E$ if and only if it is correctable for $\tr_E$ (and in particular the code is correctable for all $E$-replacer channels simultaneously). 

\begin{proof} 
The first statement is obvious from the definitions of the maps. To prove the correctable equivalence we will use the map formulation for error correction given in Eq.~(\ref{superoperatorform}).  

Suppose $S$ is correctable for $\tr_E$, so there is a channel $\mathcal R : \mathcal L(\overline{E}) \rightarrow \mathcal L(\mathcal H)$ such that $\mathcal R \circ \tr_E \circ \mathcal P_S = \mathcal P_S$. Then defining $\mathcal R^\prime = \mathcal R\circ \tr_E : \mathcal L(\mathcal H) \rightarrow \mathcal L(\mathcal H)$ yields, 
\[
%\begin{eqnarray*}
\mathcal R^\prime \circ \mathcal E_E \circ \mathcal P_S =  (\mathcal R \circ \tr_E)\circ \mathcal E_E \circ \mathcal P_S = \mathcal R \circ (\tr_E \circ \mathcal E_E) \circ \mathcal P_S = \mathcal R \circ \tr_E \circ \mathcal P_S = \mathcal P_S , 
%\end{eqnarray*} 
\]
and so $S$ is correctable for $\mathcal E_E$. 

Conversely, if $S$ is correctable for $\mathcal E_E$, then there is a channel $\mathcal R : \mathcal L(\mathcal H) \rightarrow \mathcal L(\mathcal H)$ such that $\mathcal R \circ \mathcal E_E \circ \mathcal P_S = \mathcal P_S$. Let $\mathcal F :  \mathcal L(\overline{E}) \rightarrow \mathcal L(\mathcal H)$ be the `ampliation channel' defined by 
\[
\mathcal F ( \rho_{\overline{E}} ) = \rho_{\overline{E}} \otimes \sigma_E , 
\]
where $\rho_0$ is the output state of $\mathcal D_E$ in $\mathcal E_E = \mathrm{id}_{\overline{E}}\otimes \mathcal D_E$, and note that $\mathcal E_E = \mathcal F \circ \tr_E$. Then put $\mathcal R^\prime = \mathcal R \circ \mathcal F$, and observe that, 
\[
%\begin{eqnarray*}
\mathcal R^\prime \circ \tr_E \circ \mathcal P_S =  (\mathcal R \circ \mathcal F)\circ \tr_E \circ \mathcal P_S = \mathcal R \circ (\mathcal F \circ \tr_E) \circ \mathcal P_S = \mathcal R \circ \mathcal E_E \circ \mathcal P_S = \mathcal P_S , 
%\end{eqnarray*} 
\]
and so $S$ is correctable for $\tr_E$. 
\end{proof}
\end{lem}

We next recall what in essence can be viewed as a particular manifestation of the no-cloning theorem in this setting. To state it we need to define two notions. Given a unitary $U \in \mathcal L(\mathcal H)$ 
acting on our Hilbert space $\mathcal H = \overline{E}E$, we can define a pair of {\it complementary channels} \cite{holevo2007complementary,king2005properties}: $\mathcal E_1 = \tr_E \circ \mathcal U : \mathcal L(\mathcal H) \rightarrow \mathcal L(\overline{E})$ and $\mathcal E_2 = \tr_{\overline{E}} \circ \mathcal U : \mathcal L(\mathcal H) \rightarrow \mathcal L(E)$ where $\mathcal U(\rho) = U \rho U^\dagger$.  Further, a subspace $S\subseteq \mathcal H$ is said to be a {\it private code} \cite{ambainis2000private,bartlett2005random,bartlett2004decoherence,boykin2003optimal,crann2016private,jochym2013private,kribs2021approximate,levick2017quantum} for a channel $\mathcal E: \mathcal L(\mathcal H) \rightarrow \mathcal L(\mathcal K)$ if there is a $\sigma \in \mathcal L(\mathcal K)$ such that $\mathcal E(\widetilde{\rho}) = \sigma$ for all $\widetilde{\rho}$ supported on $S$. 

As established in \cite{kretschmann2008complementarity}, private codes are complementary to error-correcting codes. We state the special case of this result that we will make use of. 

\begin{lem}\label{complemma}
    Let $U \in \mathcal L(\mathcal H)$ be a unitary operator on $\mathcal H = \overline{E}E$, and consider the complementary pair of channels $\mathcal E_1 = \tr_E \circ \mathcal U$ and $\mathcal E_2 = \tr_{\overline{E}} \circ \mathcal U$. Then a subspace $S$ of $\mathcal H$ is correctable for one of these channels if and only if it is private for the other channel.     
\end{lem}

In the proof of the structure theorem below we apply this result specifically in the case of $U=I$, and so $\mathcal U = \mathrm{id}$ is the identity map on $\mathcal L(\mathcal H)$ and the complementary pair of channels are $\tr_E$ and $\tr_{\overline{E}}$. 

As the last background result, we recall the characterization from \cite{NS06,KS06} of quantum error-correcting codes as {\it unitarily recoverable codes}. Recall that an isometry is a norm-preserving map from one Hilbert space into another. 

\begin{lem}\label{unitrecovlemma}
Let $\mathcal E : \mathcal L(\mathcal H) \rightarrow \mathcal L(\mathcal K)$ be a channel and suppose $S$ is a subspace of $\mathcal H$. Then $S$ is correctable for $\mathcal E$ if and only if there is a reference system $R\cong S$, a density matrix $\Gamma_A$ on ancilla system $A$ with $|A|= \mathrm{rank}(\Gamma) \leq \lfloor |S|^{-1}\dim\mathcal K\rfloor$, and an isometry $U: \mathbb{C}^{|S|}\otimes \mathbb{C}^{\lfloor |S|^{-1}\dim\mathcal K\rfloor} \rightarrow \mathcal K$ such that for all $\widetilde{\rho}$ supported on $S$, we have 
\begin{equation}\label{unitrecovlemmaeqn}
\mathcal E (\widetilde{\rho} ) = \mathcal U \big( \rho_R \otimes \Gamma_A \big), 
\end{equation}
where $\{ \ket{\widetilde{i}}\}$ and $\{ \ket{i}_R \}$ are bases for $S$ and $R$ respectively with $\rho_R = \op{i}{j}_R$ when $\widetilde{\rho} = \op{\widetilde{i}}{\widetilde{j}}_S$. 
\end{lem}

We note that while $U$ is an isometry in this result, in practice we can view it as a unitary by enlarging $\mc{K}$ and system $A$ such that $|S|\cdot |A|=\dim\mc{K}$ and then extending the action of $U$ to be unitary on $RA$ (this motivates the name `unitarily' recoverable). Also, from the proof of this result we know that $\mathrm{rank}(\Gamma)$ is equal to the number of syndrome subspaces for the code, which (recall from the last section) is at most the Choi rank of the channel. In the proof of the theorem below we will derive an explicit form and rank bound for the operator $\Gamma$ for replacer codes. (We will reserve use of the term `ancilla' in the theorem specifically for the Hilbert space on which this operator acts.)  
We further note that the mapping from $S$ to $R$ can be made explicit.  For $\overline{R}\cong \mbb{C}^{\lceil |S|^{-1}\dim\mc{H}\rceil}$, define the isometry $V:S\to R\overline{R}$ by
\begin{equation}
\label{Eq:iso-mapping}
V\ket{\widetilde{i}}_S = \ket{i}_R\ket{0}_{\overline{R}} ,
\end{equation}
and so $V^\dagger\ket{i}_R\ket{0}_{\overline{R}} =\ket{\widetilde{i}}_S$. 
Again, we can regard $V$ as a unitary on $\mc{H}$ by enlarging the dimension of $\mc{H}$ so that it is divisible by $|S|$ and then extending the action of $V$ to all of $\mc{H}$ appropriately.  Then whenever Eq.~(\ref{unitrecovlemmaeqn}) is satisfied,
a recovery map is given by $\mathcal R = \mathcal V_{R}^\dagger \circ \tr_A \circ \mathcal U^\dagger$, where $\tr_A$ is partial trace over the second subsystem of $\mathbb{C}^{|S|}\otimes  \mathbb{C}^{\lfloor |S|^{-1}\dim\mathcal K\rfloor}$ and $\mathcal V_R^\dagger(\cdot) = V^\dagger(\cdot\otimes \op{0}{0}_{\overline{R}})V$. 

The converse direction of the result, that correctable implies unitarily recoverable, involves a more technical proof, but as shown in \cite{NS06,KS06} the isometry/unitary $U$ can be explicitly constructed from properties of the channel and code. In Remark~\ref{computecoderemark} we describe how to compute the unitary $U$ in the case of replacer codes, from this result together with details from the proof of the theorem below.

We now state our main result as follows. 

\begin{structurethm}\label{maintheorem} 
Let $n,d\geq 1$ be positive integers and let $E$ be any set of qudits on an $n$-qudit Hilbert space $\mathcal H = \overline{E} E$. Suppose we have a subspace 
$
S = \mathrm{span}\, \{ \ket{\widetilde{i}} \} \subseteq \mathcal H . 
$
Then the following statements are equivalent:  
\begin{enumerate} 
\item[$(i)$]
$S$ is a correctable code for every replacer channel $\mathcal E_E$. 

\item[$(ii)$]
$S$ is a correctable code for some replacer channel $\mathcal E_E$. 

\item[$(iii)$]
For any replacer channel $\mathcal E_E$, there is a system $R = \mathrm{span}\, \{ \ket{i} \} \cong S$, an ancilla system $A$ with $(\dim A)( \dim S) \leq \dim \overline{E}$, and an isometry $U_{\overline{E}}: RA \rightarrow \overline{E}$ and density operator $\sigma_{A E}$ such that for all $\widetilde{\rho} = \sum_{i,j} \lambda_{ij} \kb{\widetilde{i}}{\widetilde{j}}$ supported on $S$, we have  
\begin{equation}\label{unitarilycorrectableequation}  
\mathcal E_E (\widetilde{\rho}) = \big( \mathcal U_{\overline{E}} \otimes \mathrm{id}_E \big) \, ( \rho_R \otimes \sigma_{A E}  ), 
\end{equation} 
where $\rho_R = \sum_{i,j} \lambda_{ij} \kb{i}{j}_R$ and $\mathcal U_{\overline{E}} (\cdot) = U_{\overline{E}} (\cdot)  U_{\overline{E}}^\dagger$. Further, we have 
\[
\sigma_{A E} = \Gamma_{A} \otimes \sigma_E 
\]
for some density operator $\Gamma_{A}$ and $\sigma_E$ is the $E$-output state of $\mathcal E_E$.

\item[$(iv)$]
For any replacer channel $\mathcal E_E$, there is a system $R = \mathrm{span}\, \{ \ket{i} \} \cong S$, an ancilla system $A$ with $(\dim A) (\dim S) \leq \dim \overline{E}$, and an isometry $U_{\overline{E}}: RA \rightarrow \overline{E}$ and state $\ket{\psi}_{A E}$, all determined by $\mathcal E_E$ and $S$, such that for all $\ket{\widetilde{i}}\in S$, we have 
\begin{equation}\label{stateunitaryequation}
  \ket{\widetilde{i}} = \big( U_{\overline{E}} \otimes I_E \big) \, ( \ket{i}_R \otimes \ket{\psi}_{A E}).   
\end{equation}

\item[$(v)$]
Let $\T{Q} = \mathrm{span}\, \{ \ket{i} \} \cong S$ be a reference system and define the state 
\[
\ket{\phi} = \sqrt{\dim S}^{-1} \sum_i \ket{i}_{\T{Q}} \ket{\overline{i}}_{\overline{E}E}
\]
on $\T{Q} \overline{E}E$ and let $\rho = \kb{\phi}{\phi}$. Then we have 
\begin{equation}\label{separableequation}
    \rho^{\T{Q}E} = \rho^{\T{Q}} \otimes \rho^E, 
\end{equation}
where $\rho^{\T{Q}E} = \tr_{\overline{E}}(\rho)$, $\rho^{\T{Q}} = \tr_{\overline{E}E}(\rho)$, and $\rho^{E} = \tr_{\T{Q}\overline{E}}(\rho)$.
\end{enumerate}
\end{structurethm}

\begin{proof} 
$(i) \Leftrightarrow (ii)$: The equivalence of conditions $(i)$ and $(ii)$ follows from Lemma~\ref{jointcorrectable} as one (and hence all) replacer channels $\mathcal E_E$ can correct a code $S$ if and only if the code is correctable for the partial trace $\tr_E$. 

$(ii)\Rightarrow (iii)$: Suppose that $S$ is correctable for a replacer channel $\mathcal E_E = \mathrm{id}_{\overline{E}} \otimes \mathcal D_E$ with $E$-output state $\sigma_E = \mathcal D_E(\rho_E)$. Then by Lemma~\ref{jointcorrectable} it is also correctable for $\tr_E$, and so we may apply the unitarily recoverable result Lemma~\ref{unitrecovlemma} to obtain systems $R \cong S$ and $A$ with $(\dim R)(\dim A) \leq \dim \overline{E}$, and an isometry $U_{\overline{E}}: RA \rightarrow \overline{E}$ and density operator $\Gamma_{A}$ such that for all $\widetilde{\rho}$ supported on $S$, 
\[
\tr_E ( \widetilde{\rho} ) = \mathcal U_{\overline{E}} ( \rho_R \otimes \Gamma_{A} ). 
\]
But recall that the (operator) range space of the map $\mathcal E_E$ is exactly the set of operators of the form $X_{\overline{E}} \otimes \sigma_E$, and so for all $\widetilde{\rho}$ there is a density operator $X_\rho$ on $\overline{E}$ such that $\mathcal E_E( \widetilde{\rho}) =  X_\rho \otimes \sigma_E$. Further, using the fact that $\tr_E \circ \mathcal E_E = \tr_E$, we can find this operator as:
\[
X_\rho = \tr_E (X_\rho \otimes \sigma_E ) = \big( \tr_E \circ \mathcal E_E \big) ( \widetilde{\rho} ) = \tr_E ( \widetilde{\rho} ) = \mathcal U_{\overline{E}} ( \rho_R \otimes \Gamma_{A} ).
\]
It follows that for all $\widetilde{\rho}$, we have 
\begin{eqnarray*}
    \mathcal E_E ( \widetilde{\rho} ) &=& \mathcal U_{\overline{E}} ( \rho_R \otimes \Gamma_{A} ) \otimes \sigma_E  \\ 
    &=& \Big( \mathcal U_{\overline{E}} \otimes \mathrm{id}_E \Big) \Big( \rho_R \otimes \big(  \underbrace{\Gamma_{A}  \otimes \sigma_E}_{\sigma_{A E}} \big) \Big) ,  
\end{eqnarray*}
with the operator $\sigma_{A E}$ as indicated in the last line. This establishes Eq.~(\ref{unitarilycorrectableequation}) and condition $(iii)$.

$(iii)\Rightarrow (iv)$: 
Suppose we have an isometry $U_{\overline{E}}: RA \rightarrow \overline{E}$ and density operator $\sigma_{A E}$ such that Eq.~(\ref{unitarilycorrectableequation}) holds. 
 Consider any state $\ket{\wt{\varphi}}\in S$. Since $\mc{E}_E(\op{\wt{\varphi}}{\wt{\varphi}})=\tr_E(\op{\wt{\varphi}}{\wt{\varphi}})\otimes \sigma_E$ and $\sigma_{AE}=\Gamma_A\otimes \sigma_E$, it follows that Eq. \eqref{unitarilycorrectableequation} implies 
 \[\tr_E(\op{\wt{\varphi}}{\wt{\varphi}})=\tr_{A'}[\mc{U}_{\overline{E}}(\op{\varphi}{\varphi}\otimes\op{\Gamma}{\Gamma}_{AA'})],\]
 where $\ket{\Gamma}_{AA'}$ is a purification of $\Gamma_A$.  Then by Uhlmann's theorem (which can also be viewed as a special case of the Stinespring dilation theorem), there exists an isometry $W^{(\varphi)}_{A'\to E}$ such that
 \[\ket{\wt{\varphi}}_{\overline{E}E}=U_{\overline{E}}\ket{\varphi}_R\otimes W^{(\varphi)}_{A'\to E}\ket{\Gamma}_{AA'}.\]
 We claim that $W^{(\varphi)}_{A'\to E}\ket{\Gamma}_{AA'}$ is independent of $\ket{\varphi}$.  To see this, consider any two non-orthogonal states $\ket{\varphi}$ and $\ket{\varphi'}$.  Using the previous equation to take the inner product of $\ket{\wt{\varphi}}$ and $\ket{\wt{\varphi}'}$, we find the implication 
 \[\ip{\wt{\varphi}}{\wt{\varphi}'}=\ip{\varphi}{\varphi'}\bra{\Gamma}W^{(\varphi)\dagger}W^{(\varphi')}\ket{\Gamma}\quad\Rightarrow\quad 1=\bra{\Gamma}W^{(\varphi)\dagger}W^{(\varphi')}\ket{\Gamma},\]
 which means that $W^{(\varphi)}\ket{\Gamma}=W^{(\varphi')}\ket{\Gamma}$.  But since $\ket{\varphi}$ and $\ket{\varphi'}$ are arbitrary non-orthogonal states, it follows that $\ket{\psi}_{AE}:=W^{(\varphi)}_{A'\to E}\ket{\Gamma}_{AA'}$ must be the same vector for all $\ket{\varphi}$.  Thus Eq.~(\ref{stateunitaryequation}) holds and this establishes condition $(iv)$.

$(iv)\Rightarrow (iii)$: If condition $(iv)$ holds for $\mathcal E_E$, then Eq.~(\ref{unitarilycorrectableequation}) of condition $(iii)$ holds for $\mathcal E_E$ by direct application of Eq.~(\ref{stateunitaryequation}), with the map  $\mathcal U_{\overline{E}} \otimes \mathrm{id}_E$ factoring through as, recall, $\mathcal E_E = \mathrm{id}_{\overline{E}} \otimes \mathcal D_{E}$. The decomposition of $\sigma_{A E}$ follows from the fact that the range of the map $\mathcal E_E$ consists of operators of the form $X_{\overline{E}} \otimes \sigma_E$. Explicitly, we can calculate as follows: 
\begin{eqnarray*}
\mathcal E_E ( \kb{\widetilde{i}}{\widetilde{j}} ) &=& \mathcal E_E \big(  ( \mathcal U_{\overline{E}} \otimes \mathrm{id}_E ) ( \op{i}{j}_R\otimes\op{\psi}{\psi}_{AE}  )  \big) \\ 
    &=& ( \mathcal U_{\overline{E}} \otimes \mathrm{id}_E ) (\id_{RA} \otimes\mc{D}_E)(\op{i}{j}_R\otimes\op{\psi}{\psi}_{AE}) \\ 
    &=& ( \mathcal U_{\overline{E}} \otimes \mathrm{id}_E ) (\op{i}{j}_R\otimes ( \underbrace{(\id_{A} \otimes\mc{D}_E )(\op{\psi}{\psi}_{AE})}_{\Gamma_A \otimes \sigma_E}) , 
\end{eqnarray*}
and so in particular, we have $\sigma_{AE} = \Gamma_A \otimes \sigma_E$ with 
\[
\Gamma_A=\tr_E \big( ( \mathrm{id}_A \otimes \mathcal D_E )( \op{\psi}{\psi}_{AE} ) \big), 
\]
where we note here the partial trace of $E$ is implemented over the $AE$ composite system.  

$(iii)\Rightarrow (i)$: Condition $(iii)$ is a special case of unitarily recoverable codes and, as discussed prior to the theorem, this implies that $S$ is correctable for $\mathcal E_E$ and condition $(i)$ holds. Explicitly, when $(iii)$ is satisfied, a recovery operation can be defined as 
\[
\mathcal R =  \mathcal V_R^\dagger \circ \tr_{A} \circ \mathcal U_{\overline{E}}^\dagger \circ \tr_E, 
\]
where $\mc{V}_R^\dagger(\cdot)=V^\dagger(\cdot\otimes \op{0}{0}_{\overline{R}})V$, with $V:S\to R\overline{R}$ being the isometry defined in Eq. \eqref{Eq:iso-mapping} (extended to a unitary if needed).

$(iv)\Rightarrow (v)$: An application of condition $(iv)$ is that it can be used to show that $(v)$ holds through a direct calculation (with the systems $R$ and $\T{Q}$ of the two conditions identified via an isometry). Indeed, if Eq.~(\ref{stateunitaryequation}) holds for some choice of replacer channel $\mathcal E_E$, then for all $i$ we have 
\begin{eqnarray*}
\tr_{\overline{E}} ( \kb{\widetilde{i}}{\widetilde{i}}_{\overline{E}E} ) &=& 
    \tr_{\overline{E}} \big( (U_{\overline{E}} \otimes I_E ) \big( \kb{i}{i}_R \otimes \kb{\psi}{\psi}_{A E} \big)  (U_{\overline{E}}^\dagger  \otimes I_E ) \big) \\ 
      &=& \tr_{A}  \big( \kb{\psi}{\psi}_{A E} \big), 
\end{eqnarray*}
where we have used the fact that $U_{\overline{E}} : R A \rightarrow \overline{E}$. 
Hence we find for $\rho = \kb{\phi}{\phi}$, 
\begin{eqnarray*}
\rho^{\T{Q}E} = \tr_{\overline{E}}(\rho) &=& \frac{1}{\dim S} \sum_{i,j} \kb{i}{j}_{\T{Q}} \otimes \tr_{\overline{E}} ( \kb{\widetilde{i}}{\widetilde{j}}_{\overline{E}E} ) \\ 
&=& \frac{1}{\dim S} \sum_{i} \kb{i}{i}_{\T{Q}} \otimes \tr_{\overline{E}} ( \kb{\widetilde{i}}{\widetilde{i}}_{\overline{E}E} ) \\
&=& \Big( (\dim S)^{-1} I_{{\T{Q}}} \Big) \otimes \Big(  \tr_{A} \big( \kb{\psi}{\psi}_{AE} \big) \Big) .  
\end{eqnarray*}
On the other hand, we have 
\[
\rho^{\T{Q}} = \tr_{\overline{E}E}(\rho) = \frac{1}{\dim S} \sum_{i,j} \kb{i}{j}_{\T{Q}} \, \tr \big(  \kb{\widetilde{i}}{\widetilde{j}}_{\overline{E}E}   \big) = \frac{1}{\dim S} I_{\T{Q}}, 
\]
and 
\[
\rho^{E} = \tr_{\T{Q}\overline{E}}(\rho) = \frac{1}{\dim S} \sum_i \tr_{\overline{E}} \big( \kb{\widetilde{i}}{\widetilde{i}}_{\overline{E}E} \big) = \tr_{A}  \big( \kb{\psi}{\psi}_{AE} \big). 
\]
It follows that $\rho^{\T{Q}E} = \rho^{\T{Q}} \otimes \rho^E$ and condition $(v)$ holds. 

$(v)\Rightarrow (i)$: Finally, to complete the proof we show that Eq.~(\ref{separableequation}) implies that $S$ is correctable for any $E$-replacer channel. To do so, we will use the complementarity result of Lemma~\ref{complemma}. We can start by using the calculations of the previous paragraph (those that do not depend on condition $(iv)$) to show that for all $k$, with $P_{k} = \kb{k}{k}_{\T{Q}} \otimes I_E$, 
\begin{eqnarray*} 
P_{k} (\rho^{\T{Q}E}) P_{k} &=& P_{k} \Big( \frac{1}{\dim S} \sum_{i,j} \kb{i}{j}_{\T{Q}} \otimes \tr_{\overline{E}} ( \kb{\widetilde{i}}{\widetilde{j}}_{\overline{E}E} )   \Big) P_{k} \\
&=&  \kb{k}{k}_{\T{Q}} \otimes\Big( \frac{1}{\dim S} \tr_{\overline{E}} ( \kb{\widetilde{k}}{\widetilde{k}}_{\overline{E}E} ) \Big)  , 
\end{eqnarray*} 
and 
\begin{eqnarray*}
P_k ( \rho^{\T{Q}} \otimes \rho^E ) P_k &=& \Big( \frac{1}{\dim S} \kb{k}{k}_{\T{Q}} \Big) \otimes \Big( \frac{1}{\dim S} \sum_i \tr_{\overline{E}} \big( \kb{\widetilde{i}}{\widetilde{i}}_{\overline{E}E} \big)  \Big) \\ &=&   
\kb{k}{k}_{\T{Q}} \otimes \Big( \frac{1}{\dim S} \tr_{\overline{E}} \big(\frac{1}{\dim S} P_S \big) \Big), 
\end{eqnarray*}
where $P_S$ is the projection of $\mathcal H$ onto $S$. Thus when Eq.~(\ref{separableequation}) holds, for all $i$ we have  
\[
\tr_{\overline{E}} ( \kb{\widetilde{i}}{\widetilde{i}}_{\overline{E}E} ) = \tr_{\overline{E}} \big(\frac{1}{\dim S} P_S \big). 
\]
It follows that for every density operator $\widetilde{\rho} = \sum_{i,j} p_{ij} \kb{\widetilde{i}}{\widetilde{j}}_{\overline{E}E}$ supported on $S$, we have 
\[
\tr_{\overline{E}} ( \widetilde{\rho}  ) = \sum_{i,j} p_{ij} \tr_{\overline{E}} \big( \kb{\widetilde{i}}{\widetilde{j}}_{\overline{E}E} )  \big) = \sum_{i} p_{ii} \tr_{\overline{E}} \big( \kb{\widetilde{i}}{\widetilde{i}}_{\overline{E}E} )  \big) = \tr_{\overline{E}} \big(\frac{1}{\dim S} P_S \big). 
\]
Hence $S$ is a private subspace for the partial trace map $\tr_{\overline{E}}$. By Lemma~\ref{complemma}, it follows that $S$ is correctable for the complementary trace map $\tr_{E}$, which in turn is equivalent to being correctable for any $E$-replacer channel by Lemma~\ref{jointcorrectable}, and so $(i)$ holds. This completes the proof. 
\end{proof}

Before moving on to the next section, some remarks are in order on the conditions in the theorem statement and the content of the proof.

\begin{remark}
See Figure~1 for a diagrammatic circuit perspective of the conclusions of the theorem in the general case. We note that many of the naturally arising examples of replacer codes (in particular those that arise in secret sharing schemes) are codes that can be directly identified with a subsystem of the system Hilbert space; i.e., the code dimension divides the overall Hilbert space dimension. As in most of the examples presented below, in such cases we can identify $S \cong \overline{E}_1$ where  $\overline{E} = \overline{E}_1 \overline{E}_2$ (following the notation of \cite{AXH15} for instance), and $\overline{E}_2$ is potentially some other subsystem that does not get erased. Recall that we always have $\dim S \leq \dim \overline{E}$, so this is possible precisely when the code dimension divides the dimension of $\mathcal H = \overline{E} E$ (and of $\overline{E}$). However, unlike the case in previous examples, our theorem identifies subsystems $R, A$ that are not necessarily equivalent to the given qudit subsystem decomposition. (This point is elucidated further in Example~\ref{ex:grassl-four-qubit}.) 
\end{remark}

\begin{figure}[bt!]
    \centering
    \includegraphics[width=0.95\linewidth]{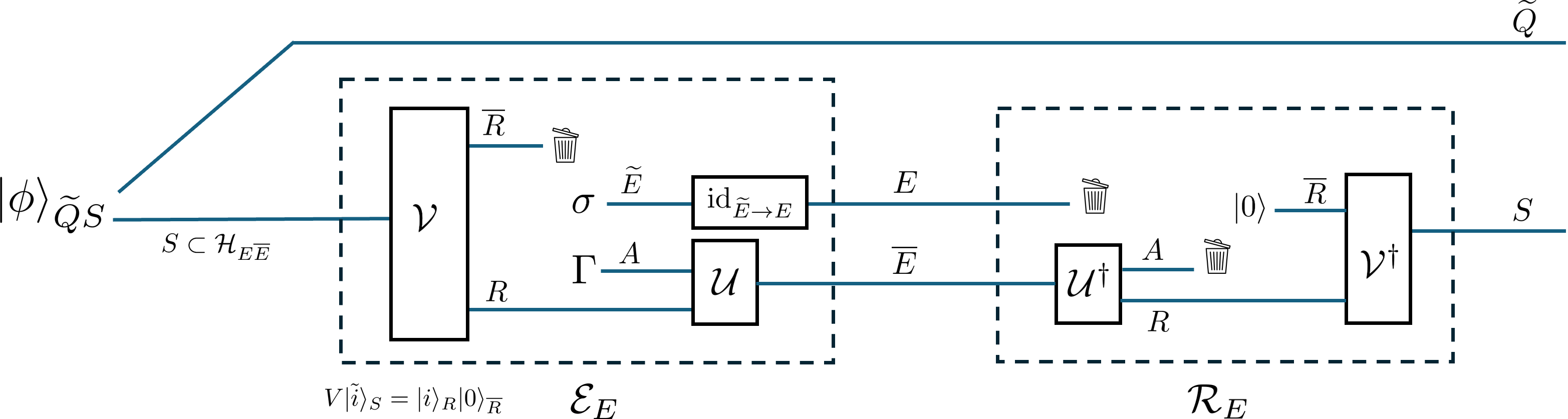}
    \caption{\textbf{Conditions for correctability of a replacer on part of an encoded state}: Circuit diagram of the structure theorem of a $|S|$-dimensional quantum replacer code for joint systems $E\overline{E}$.  By extending systems as described in the text following Lemma \ref{unitrecovlemma}, the maps $\mc{U}$ and $\mc{V}$ apply unitary transformations $U$ and $V$, respectively.}
    \label{fig:maintheorem}
\end{figure}

\begin{remark}
Let us expand on the information-theoretic viewpoint of the theorem. In particular, note that condition $(v)$ is equivalent to the mutual information $I(\T{Q}:E)_{\ket{\phi}}$ of the states vanishing; that is, 
\begin{equation}\label{mutualinformationequation}
    I(\T{Q}:E)_{\ket{\phi}} = H(\T{Q})_{\ket{\phi}} + H(E)_{\ket{\phi}} - H(\T{Q}E)_{\ket{\phi}} = 0, 
\end{equation}
where $S(\T{Q})_{\ket{\phi}} = - \tr (\rho^{\T{Q}} \log (\rho^{\T{Q}}) ),$ etc., is the entropy of the reduced state $\rho^{\T{Q}}=\tr_{E\overline{E}}\op{\phi}{\phi}$. This follows from well-known properties of quantum relative entropy.  

Further, from Figure~\ref{fig:maintheorem}, we see that the full channel $\mc{E}_E$ can be expressed as
\begin{align}
    \mc{E}_E(\rho_{E\overline{E}})=\tr_{R'A'E'}[\id_{\wt{E}\to E}\circ\mc{U}\circ\mc{V}(\rho_{E\overline{E}}\otimes\op{\Gamma}{\Gamma}_{AA'}\otimes\op{\sigma}{\sigma}_{\wt{E}E'})]
\end{align}
where $V:E\overline{E}\to R\overline{R}$ and $U:RA\to\overline{E}$ are unitaries, while $\ket{\Gamma}_{AA'}$ and $\ket{\sigma}_{\wt{E}E'}$ are purifications of $\Gamma_A$ and $\sigma_{\wt{E}}$, respectively (the purifying systems are omitted in Fig.~\ref{fig:maintheorem}).  This is a standard `Stinespring representation' of the channel in which the unitary $\mc{U}\circ\mc{V}$ is being applied on the larger system consisting of $E\overline{E}$ and the auxiliary systems $AA'\wt{E}E'$.  The map $\id_{\wt{E}\to E}$ simply transfers the prepared state $\sigma_{\wt{E}}$ from register $\wt{E}$ to the system register $E$. 
Notice that when $\rho_{E\overline{E}}$ has support on $S$, we have $\mc{V}(\rho_{E\overline{E}})=\rho_R\otimes\op{0}{0}_{\overline{R}}$ from Eq.~\eqref{Eq:iso-mapping}.  Thus, the final state in the purified picture is given by 
\[\ket{\Psi}=UV\ket{\Phi^+}_{S\widetilde{Q}}\ket{\Gamma}_{AA'}\ket{\sigma}_{\wt{E}E}=U\ket{\Phi^+}_{R\widetilde{Q}}\ket{\Gamma}_{AA'}\ket{\sigma}_{EE'}\ket{0}_{\overline{R}},\]
and the entropy of systems $E\overline{E}$ is
\begin{align}
H(E\overline{E})_{\ket{\Psi}}&=H(S)_{\ket{\Phi^+}}+H(A)_{\ket{\Gamma}}+H(E)_{\ket{\sigma}}=H(S)_{\ket{\Phi^+}}+H(A')_{\ket{\Gamma}}+H(E')_{\ket{\sigma}}.
\end{align}
This says that the final entropy of system $E\overline{E}$ is equal to the initial entropy of system $S$ plus the final entropy of the auxiliary systems.  As shown by Schumacher and Nielsen \cite{SN96}, this condition is necessary and sufficient for $S$ to be correctable for $\mc{E}_E$.

\end{remark}

\begin{remark}\label{partialreplacernote} 
While the theorem is stated for replacer channels that are endomorphic on the space $\mc{H} = \overline{E}E$ (as in Definition~\ref{def:q-replacer-channels}), we can easily generalize the result to replacer channels built with a private channel $\mc{D}_E$ whose output space is different than $E$.  This includes the case of so-called `flagged' erasure channels, that replace the input with a flag state that is orthogonal to the support space of the input.  

In fact, we can also generalize the theorem to include channels that erase the input with only some probability $\lambda\in[0,1)$.  We define a \textit{partial replacer channel on $E$} as any channel of the form
\begin{equation}
    \mc{E}_{\lambda,\sigma}(\rho_{E\overline{E}})=\lambda \, \rho_{E\overline{E}}+(1-\lambda)\, \sigma_{E'}\otimes\tr_E (\rho_{E\overline{E}}),
\end{equation}
where $\sigma_{E'}$ is an arbitrary (and fixed) state of some system $E'\supseteq E$. 
    \begin{lem}
    A code is correctable for any partial replacer channel if and only if it is correctable for the completely depolarizing channel on $E$.
\end{lem}

\begin{proof}
    Suppose a code $S$ is correctable for a partial replacer channel $\mc{E}_{\lambda,\sigma}$.
    Let 
    \[
    Z(k)=\sum_{j=0}^{|E|-1}\omega^{jk}\op{j}{j} \quad \mathrm{and} \quad  X(l)=\sum_{j=0}^{|E|-1}\op{j\oplus l}{l}
    \]
    be the generalized Pauli operators on system $E$, where $\omega=e^{2\pi i/|E|}$ and $\oplus$ denotes addition modulo $|E|$.  Then a set of Kraus operators for $\mc{E}_{\lambda,\sigma}$ are given by 
    \[
    \big\{\sqrt{\lambda} I_{\overline{E}E}, \, I_{\overline{E}}\otimes\sqrt{(1-\lambda)/|E|}\sqrt{\sigma}X(l)Z(k)\big\}_{l,k=0}^{|E|-1}.  
    \]
    If $\ket{\widetilde{i}}$ denotes orthonormal basis vectors for $S$, then by the Knill-Laflamme conditions, correctability of $S$ implies that these Kraus operators satisfy $\bra{\wt{i}}M^\dagger_aM_b\ket{\wt{j}}= c_{ab}\delta_{ij}$
    for some constants $c_{ab}$.  In particular, we have
    \begin{align}
        \frac{1-\lambda}{|E|}\bra{\wt{i}}\mbb{I}_{\overline{E}}\otimes Z(k)^\dagger X(l)^\dagger \sigma X(l)Z(k)\ket{\wt{j}}=c_{lk}\delta_{ij} \qquad \forall l,k=0,\cdots,|E|-1.
    \end{align}
By defining $c_{lk}'=\frac{|E|}{1-\lambda}c_{lk}$, we see that the Knill-Laflamme conditions for the completely replacer channel $\mc{E}_{0,\sigma}$ are satisfied with the structure constants $c'_{lk}$.  Hence, $S$ is correctable for $\mc{E}_{0,\sigma}$ and also for the completely depolarizing channel by the structure theorem.

For the converse, suppose that a code $S$ is correctable for the completely depolarizing channel on $E$.  Then by the theorem, an orthonormal basis for $S$ is given by  $\ket{\wt{i}}_{\overline{E}E}=U_{RA\to\overline{E}}\ket{i}\ket{\psi}_{AE}$ for some isometry $U_{RA\to \overline{E}}$.  But if $\{I_{\overline{E}}\otimes M_a\}_a$ is \textit{any} set of error operators with $M_a:E\to E'$, we then have
\begin{align}
    \bra{\wt{i}}M_a^\dagger M_b\ket{\wt{j}}=\delta_{ij}\bra{\psi} I_A\otimes M_a^\dagger M_b\ket{\psi}=\delta_{ij}c_{ab}.
\end{align}
Therefore, $S$ is correctable for any channel having Kraus operators of the form $\{I_{\overline{E}}\otimes M_a\}_a$, which includes all $E$-replacer channels.
\end{proof}

It is perhaps useful to point out the following conceptually obvious fact that follows directly from the proof of this result. 

\begin{cor}
    Completely depolarizing (or replacer) channels are the most destructive in the sense that being able to correct for complete depolarization (or replacer) on some subsystem enables the ability to correct for any other noise on that subsystem.
\end{cor}

\end{remark}

\begin{remark}\label{computecoderemark}
Finally, as a lead-in to the applications and examples of the next section, we describe how the structure theorem, together with its preceding lemmas and proof details, provides a way to construct the unitary operators $U_{\overline{E}}$ and states $\Gamma_A$ and $\ket{\psi}_{A E}$ that implement these codes. The key point is that we can use the unitarily recoverable description of correctable codes for the partial trace operations to do this.  

Suppose $S \subseteq \mathcal H = \overline{E}E$ is a correctable code for a replacer channel $\mathcal{E}_E$ on $\mathcal H$. Then, by Lemma~\ref{jointcorrectable}, it is correctable for the partial trace operation $\tr_E$ that traces out system $E$. 
In particular, from Lemma~\ref{unitrecovlemma} we have for all $\widetilde{\rho}$ supported on $S$, 
\begin{equation}\label{unitaryconstruct1}
\tr_E (\widetilde{\rho}) = \mathcal{U}_{\overline{E}} (\rho_R \otimes \Gamma_A), 
\end{equation}
for some state $\rho_R$ determined by $\widetilde{\rho}$ on a reference system $R$ isometric to $S$ and fixed ancilla state $\Gamma_A$ with its rank constraint given in the lemma. 

Before finding a unitary, we can start by determining a state $\Gamma_A$ (and also $\dim A$) by simply evaluating the left-hand side of Eq.~(\ref{unitaryconstruct1}) for $\widetilde{\rho}$ equal to any rank-one projection on $S$. As $\rho_R$ is also rank-one and $U_{\overline{E}}$ is an isometry, the rank of $\tr_E (\widetilde{\rho})$ is equal to the rank of $\Gamma_A$, which we can define as a full-rank operator and hence this rank is also equal to $\dim A$. More than this, we can explicitly find a suitable operator $\Gamma_A$ by computing the spectral decomposition of $\tr_E (\widetilde{\rho})$ and then using its non-zero eigenvalues $\{ \gamma_k \}_k$ (including multiplicities) together with any orthonormal basis $\{\ket{k}_A \}_k$ for $A$ to define $\Gamma_A = \sum_k \gamma_k \kb{k}{k}_A$.

We can then obtain $U_{\overline{E}}$ by evaluating Eq.~(\ref{unitaryconstruct1}) with $\widetilde{\rho} = \kb{\widetilde{i}}{\widetilde{j}}$ for any orthonormal basis $\{\ket{\widetilde{i}}\}_i$ for the code space $S$. 
In summary, we do the following: 
\begin{enumerate}
    \item Choose an orthonormal basis $\{\ket{j}_R\}_{j\in \mathcal{J}}$ for the reference system $R$.
    \item Choose an orthonormal basis $\{\ket{k}_A\}_k$ for the ancilla system $A$ and define a density operator $\Gamma_A = \sum_k \gamma_k \kb{k}{k}_A$ as discussed above.  
    \item The isometry $U_{\overline{E}}$ is then obtained by evaluating Eq.~(\ref{unitaryconstruct1}) with $\widetilde{\rho}$ equal to the rank-one operators determined by any orthonormal basis for $S$.  
\end{enumerate}
We will see this unitary and code construction explicitly in the examples of the next section.

\end{remark}

\section{Gallery of Applications and Examples}\label{sec:gallery}

In this section, we present a number of applications and examples for specific classes of replacer codes, some of which are drawn from different investigations found in the literature. We indicate how they all can be viewed from the overarching perspective of the structure theorem, by indicating the unitary operator and states given by the conditions of the theorem. We also show the utility of the result by exhibiting how new codes and results, and alternate proofs of certain results, can be derived from it. 

We have divided the section into subsections to help organize the presentation, but the classes of codes described in each subsection are not mutually exclusive. In addition, to further streamline the presentation, we have left out the normalization constants when defining and discussing the code basis states in each of the examples.

\subsection{Trivial Replacer Codes}
We begin by describing a simple class of codes, those defined by subspaces of the untouched qudits together with a fixed state on the replaced qudits. 

\begin{example}
Given any replacer channel $\mathcal E_E = \mathrm{id}_{\overline{E}} \otimes \mathcal D_E$ on $\mathcal H = \overline{E} E$, with $\mathcal D_E(\rho_E) = \sigma_E$ for all $\rho_E$, one always has its {\it trivial correctable codes} given as follows. Fix a state $\ket{\psi}_E \in E$ and a subspace $S^\prime$ of $\overline{E}$, and define the subspace $S = S^\prime \otimes \ket{\psi}_E$ of $\mathcal H$. If $\mathcal R = \mathrm{id}_{\overline{E}} \otimes \mathcal F_E$ where $\mathcal F_E$ is any channel on $E$ with $\mathcal F_E(\sigma_E) = \kb{\psi}{\psi}_E$, then $(\mathcal R \circ \mathcal E_E)(\widetilde{\rho}) = \widetilde{\rho}$ for all $\widetilde{\rho}$ supported on $S$, and so $S$ is a correctable code for $\mathcal E_E$. In the context of the structure theorem, note that here we can take $A=\mathbb C$ and a unitary $U_{\overline{E}}: R \rightarrow S^\prime \subseteq \overline{E}$, to obtain code basis states as $\ket{\widetilde{i}} = ( U_{\overline{E}} \otimes I_E ) \, ( \ket{i}_R \otimes \ket{\psi}_{E})$.

In fact, in light of the theorem, the trivial codes can be seen to correspond exactly to the cases in which the state $\ket{\psi}_{AE}$ is separable across the combined $AE$ system. Indeed, if $S$ is a correctable code for $\mathcal E_E$ and $\ket{\psi}_{AE} = \ket{\phi}_A \ket{\phi}_E$ (where we allow the $A$ state to simply be the number $1$ when $A=\mathbb C$), then a basis for $S$ is given by 
\[
\ket{\widetilde{i}} = \big( \underbrace{U_{\overline{E}} \,  \ket{i}_R \ket{\phi}_A}_{\in \overline{E}} \big) \, \ket{\phi}_{E}.
\]
Hence, $S = S^\prime \otimes \ket{\phi}_E$ is of the stated form with $S^\prime = U_{\overline{E}} (R \otimes \ket{\phi}_A)$. 
\end{example}

The trivial codes for specific replacer channels are thus easy to obtain and characterize. Of course, they are not of much use for applications, as, for instance, such a code is in general not correctable for any replacer channel that disrupts the subsystems that are left unchanged by the original channel.

\subsection{Quantum Erasure Codes}

Next we revisit one of the first non-trivial analyses of quantum erasure codes, with an example and then an alternate proof of a no-go result based on the theorem. 

\begin{example}\label{ex:grassl-four-qubit}
Quantum erasure codes were investigated by Grassl, et al, in \cite{grassl1997codes}. In our terminology, recall these are the replacer codes for which the private map is the completely depolarizing channel. We first show how a motivating example from \cite{grassl1997codes} can be viewed from the structure theorem perspective. Consider the two-qubit code $S$ on four-qubit space with (unnormalized) basis states given by: 
    \begin{align}
        \ket{\tilde{0}} = \ket{0000}+\ket{1111} &  \quad \quad      \ket{\tilde{2}} = \ket{1100}+\ket{0011} \notag \\
       \ket{\tilde{1}} = \ket{1001}+\ket{0110}  & \quad \quad  \ket{\tilde{3}} = \ket{1010}+\ket{0101} \notag \\
       \notag
    \end{align}  
    
This code is correctable for each of the four single qubit erasures. Let us observe the code from the perspective of the theorem. In the case of erasure of the fourth qubit, for instance, we have $\mathcal H = \overline{E} E = ( (\mathbb C^2)^{\otimes 3} ) \otimes \mathbb C^2$, with the reference system $R = \mathrm{span} \, \{ \ket{0}_R, \ket{1}_R, \ket{2}_R, \ket{3}_R  \} \cong S$. We can also immediately say the ancilla $A$ is either $\mathbb C$ or $\mathbb C^2$ as $4 (\dim A) = (\dim R) (\dim A) \leq \dim \overline{E} =8$. As in the proof of the theorem (and Remark~\ref{computecoderemark}), we can discern $\Gamma_A$ and the action of the unitary $U_{\overline{E}}: RA \rightarrow \overline{E}$ via the equation 
\[
\tr_4 (\widetilde{\rho}) = \mathcal U_{\overline{E}} (\rho_R \otimes \Gamma_A).  
\]

By choosing $\widetilde{\rho} = \kb{\Tilde{0}}{\Tilde{0}}$, we find 
\[
\tr_4 ( \kb{\tilde{0}}{\tilde{0}} ) = \frac12 ( \kb{000}{000} + \kb{111}{111}  ).
\]
Hence $A = \mathbb C^2$ and $\{ \frac12 , \frac12 \}$ are the (non-zero) eigenvalues of the ancilla state; that is, $\Gamma_A = \frac12 I_2$. 
Calculating, we find we can define the unitary $U_{\overline{E}}$ via its action on basis states as: 
\[
\begin{array}{cccc} 
R & A & U_{\overline{E}}  & \overline{E}  \\ \hline 
\ket{0} & \left\{ \begin{array}{l}  \ket{0} \\ \ket{1}  \end{array} \right. & \longmapsto & \left\{ \begin{array}{l}  \ket{000} \\ \ket{111}  \end{array} \right. \\ 
\ket{1} & \left\{ \begin{array}{l}  \ket{0} \\ \ket{1}  \end{array} \right. & \longmapsto & \left\{ \begin{array}{l}  \ket{011} \\ \ket{100}  \end{array} \right. \\  
\ket{2} & \left\{ \begin{array}{l}  \ket{0} \\ \ket{1}  \end{array} \right. & \longmapsto & \left\{ \begin{array}{l}  \ket{110} \\ \ket{001}  \end{array} \right. \\  
\ket{3} & \left\{ \begin{array}{l}  \ket{0} \\ \ket{1}  \end{array} \right. & \longmapsto & \left\{ \begin{array}{l}  \ket{101} \\ \ket{010}  \end{array} \right. 
\end{array} 
\]

In particular, for each $0 \leq i \leq 3$ one can verify that we have 
\[
  \ket{\widetilde{i}} = \big( U_{\overline{E}} \otimes I_E \big) \, ( \ket{i}_R \otimes \ket{\psi}_{A E}),  
\]
where here $\ket{\psi}_{AE} = \frac{1}{\sqrt{2}} (\ket{00} + \ket{11})$ is the canonical maximally entangled state between the ancilla and erased qubit (the purification of $\Gamma_A$). For example, with $i=2$ we have 
\[
 \big( U_{\overline{E}} \otimes I_E \big) \, ( \ket{2}_R \otimes \ket{\psi}_{A E}) = \frac{1}{\sqrt{2}} \big( U_{\overline{E}} \otimes I_E \big) \, \big( \ket{2}_R \otimes   ( \ket{0}_A \otimes \ket{0}_E +  \ket{1}_A \otimes \ket{1}_E ) \big) =   \ket{\widetilde{2}} .   
\]

As an indication of the utility of the theorem, note that any state $\ket{\psi}_{AE}$ on $AE$ (and a fixed isometry $U_{\overline{E}}$) defines a code via the codeword condition~$(iv)$ for erasure of the fourth qubit. One could change it to another maximally entangled state, which would change the code words but maintain the correctability of the code for all single qubit erasures (as maximally entangled states on $AE$ are related by a local unitary on $A$, which can be absorbed into the isometry).  Using a separable state would also give a different correctable code, and give trivial codes in the sense above in that they are directly determined by subspaces of $\overline{E}$, though such codes would no longer be correctable for the other qubit erasures. One could also work with the operator $\Gamma_A$ to define different codes. For instance, with the same isometry and putting $\Gamma_A = p_1 \kb{0}{0} + p_2 \kb{1}{1}$ instead of the maximally mixed state, we would keep the same basis states that define the code states $\ket{\widetilde{i}}$ with coefficients changed to $\sqrt{p_1}$, $\sqrt{p_2}$ above (instead of $\frac{1}{\sqrt{2}}, \frac{1}{\sqrt{2}}$). 

Further, one could consider different isometries $U_{\overline{E}}$ and generate different codes. For instance, we could keep $A= \mathbb C^2$ and $\Gamma_A = \frac12 I_2$ (and $\ket{\psi}_{AE}$ equal to the maximally entangled state above), and define a qubit code $S \cong R = \mathbb C^2$ with the isometry defined as: 
\[
\begin{array}{cccc} 
R & A & U_{\overline{E}}  & \overline{E}  \\ \hline 
\ket{0} & \left\{ \begin{array}{l}  \ket{0} \\ \ket{1}  \end{array} \right. & \longmapsto & \left\{ \begin{array}{l}  \ket{000} + \ket{111} \\ \ket{000} - \ket{111}  \end{array} \right. \\ 
\ket{1} & \left\{ \begin{array}{l}  \ket{0} \\ \ket{1}  \end{array} \right. & \longmapsto & \left\{ \begin{array}{l}  \ket{011} + \ket{100}  \\ \ket{011} - \ket{100}   \end{array} \right.  
\end{array} 
\]
This would define the correctable qubit code via Eq.~(\ref{stateunitaryequation}) defined by the following basis states: 
    \begin{align}
        \ket{\tilde{0}} &=  \ket{0000}+\ket{1110} + \ket{0001} - \ket{1111}  &   \notag \\
       \ket{\tilde{1}} &=  \ket{0110}+\ket{1000} + \ket{0111} - \ket{1001}  \notag \\
       \notag
    \end{align} 
By adding an extra dimension to make $R$ a qutrit, and extending the isometry to $\ket{2}\otimes A$, one obtains a qutrit correctable code for erasure of the fourth qubit, and a simple example of a code with dimension that does not divide the dimension of the overall Hilbert space ($\dim S =3$, $\dim \mathcal H = 8$). Another way to define a simple qutrit code example here would be to restrict the original two-qubit unitary above to a three-dimensional subspace. (A more interesting example of this dimensionality phenomena is given in Example~\ref{nondividingdimensioneg}.)    
\end{example}

We can also use the structure theorem to give alternate proofs of previous results. For instance, Theorem~1 of \cite{grassl1997codes} asserts that if a correctable erasure code has a single state that has a tensor factor belonging to the erased qubits, then every state in the code must have that same state factor. This is a key result in the building of other results in \cite{grassl1997codes} and subsequent works, and we can obtain this as a straightforward consequence of the code state description of the theorem. 

\begin{cor}[Theorem~1 of \cite{grassl1997codes}]
Let $S$ be a correctable code for an erasure channel $\mathcal E_E$ on $\mathcal H = \overline{E}E$. If there is a state $\ket{\widetilde{\psi}} \in S$ for which there is a state  $\ket{\psi_0}$ on a subset $E_2$ of the qudits $E=E_1E_2$ such that $\ket{\widetilde{\psi}} = \ket{\psi^\prime}\ket{\psi_0}$ for some state $\ket{\psi^\prime}$ on $\overline{E}E_1$, then $\ket{\psi_0}$ is a factor of every state in $S$; that is, $S = S^\prime \otimes \ket{\psi_0}$ where $S^\prime$ is a subspace of $\overline{E}E_1$. 
\end{cor}

\begin{proof} 
In light of the theorem, this can seen to be a relic of condition~$(iv)$; namely, if a code state has a tensor factor over the erased qubits, then the fixed state $\ket{\psi}_{AE}$ is separable with that same tensor factor, and hence so are all of the code states. 
\end{proof}

\subsection{Secret Sharing Seminal Example}

As a further illustration of the structure theorem, we next revisit one of the most referenced works on quantum erasure codes and a key starting point for quantum secret sharing scheme investigations. 

\begin{example}
An important example of a threshold secret sharing scheme introduced by Cleve, et al, in 
\cite{cleve1999share} encodes a qutrit into three-qutrit space with the following (unnormalized) basis states: 
    \begin{align}
        \ket{\widetilde{0}} = \ket{000} + \ket{111} + \ket{222}  \notag \\
       \ket{\widetilde{1}} = \ket{012} + \ket{120} + \ket{201} \notag \\ 
       \ket{\widetilde{2}} = \ket{021} + \ket{102} + \ket{210} \notag
    \end{align}

The qutrit code $S= \mathrm{span}\, \{ \ket{\widetilde{0}}, \ket{\widetilde{1}}, \ket{\widetilde{2}} \}$ spanned by the three encoded states above is a correctable code on $\mathcal H = (\mathbb C^3)^{\otimes 3}$ for each of the three individual qutrit erasure maps (but not any two erasures simultaneously, hence the threshold designation of the code). This particular code has been widely referred to and often used for illustrative purposes, including more recently in the context of black hole theory \cite{AXH15}, where the unitary and information-theoretic forms for this code (i.e., conditions $(iv)$ and $(v)$) were identified, and then building on that work in \cite{10206477} as a foundation for analysis of tripartite secret sharing schemes. 

As the unitary that corresponds to erasing the third qutrit was explicitly given in \cite{AXH15}, here let us give a unitary determined by erasure of the first qutrit. In this case we have $\mathcal H = E \overline{E} = \mathbb C^3 \otimes ( (\mathbb C^3)^{\otimes 2} )$ and the reference system $R = \mathrm{span} \, \{ \ket{0}_R, \ket{1}_R, \ket{2}_R \} \cong S$. To find the size of $A$ and $\Gamma_A$, as in the previous example we can choose $\widetilde{\rho} = \kb{\widetilde{0}}{\widetilde{0}}$ and use the relation:  
\[
\mathcal U_{\overline{E}} (\kb{0}{0}_R \otimes \Gamma_A) = \tr_1 (\kb{\widetilde{0}}{\widetilde{0}} )  = \frac13 \big( \kb{00}{00} +   \kb{11}{11} + \kb{22}{22}  \big) . 
\]
Hence, we take $A = \mathbb C^3$ and $\Gamma_A = \frac13 I_3$ here as we have $\mathrm{rank}(\Gamma_A)=3$ and its eigenvalues are $\{ \frac13, \frac13, \frac13 \}$.  
One can then define the unitary action on basis states by direct calculation, given as follows: 
\[
\begin{array}{cccc} 
R & A & U_{\overline{E}}  & \overline{E}  \\ \hline 
\ket{0} & \left\{ \begin{array}{l}  \ket{0} \\ \ket{1} \\ \ket{2} \end{array} \right. & \longmapsto & \left\{ \begin{array}{l}  \ket{00} \\ \ket{11} \\ \ket{22} \end{array} \right. \\ 
\ket{1} & \left\{ \begin{array}{l}  \ket{0} \\ \ket{1} \\ \ket{2} \end{array} \right. & \longmapsto & \left\{ \begin{array}{l}  \ket{12} \\ \ket{20} \\ \ket{01}  \end{array} \right. \\  
\ket{2} & \left\{ \begin{array}{l}  \ket{0} \\ \ket{1} \\ \ket{2} \end{array} \right. & \longmapsto & \left\{ \begin{array}{l}  \ket{21} \\ \ket{02} \\ \ket{10}  \end{array} \right.  
\end{array} 
\]
Then, again, one can verify here that for each $0 \leq i \leq 2$ we have 
\[
  \ket{\widetilde{i}} = \big( U_{\overline{E}} \otimes I_E \big) \, ( \ket{i}_R \otimes \ket{\psi}_{A E}),  
\]
where in this case $\ket{\psi}_{AE} = \frac{1}{\sqrt{3}} (\ket{00} + \ket{11} + \ket{22})$ is the purification of $\Gamma_A$. 
\end{example}

As above, we can also use the structure theorem to find alternative proofs of results, and in this case, we find the theorem gives us additional information on a key result from \cite{cleve1999share,gottesman2000theory}. The following description of correctable erasure codes was uncovered in \cite{cleve1999share}, and also stated as Theorem~1 in \cite{gottesman2000theory}. 

\begin{cor}\label{correctablecor}
    Let $S$ be a subspace of $\mathcal H = \overline{E} E$. Then $S$ corrects erasure errors on $E$ if and only if for every operator $X = I_{\overline{E}} \otimes X_E$ on $\mathcal H$ that only acts non-trivially on $E$, there is a scalar $c(X)$ such that 
    \[
    \bra{\widetilde{\phi}} X \ket{\widetilde{\phi}} = c(X), 
    \]
    for all $\widetilde{\phi}\in S$. Further, when this condition is satisfied we have 
    \[
    c(X) =  \bra{\psi} \, (I_A \otimes X_E) \,   \ket{\psi} , 
    \]
    where $A$ is the ancilla and $\ket{\psi} = \ket{\psi}_{AE}$ is the state from condition~$(iv)$ of the structure theorem. 
\end{cor}

\begin{proof} 
The description of a correctable erasure code $S$ in terms of the scalars $c(X)$ is a straightforward application of the Knill-Laflamme Theorem \cite{knill1997theory}, as noted in \cite{cleve1999share,gottesman2000theory}, with the key observation being that one can obtain sets of Kraus operators for erasure channels that form a multiplicatively closed set up to scalar multiples. To identify the scalars $c(X)$, we can use the explicit code state form given in Eq.~(\ref{stateunitaryequation}) of the theorem to obtain them as follows: for all $\ket{\widetilde{\phi}}\in S$ and $X = I_{\overline{E}} \otimes X_E$, we have 
\begin{eqnarray*}
     \bra{\widetilde{\phi}} X \ket{\widetilde{\phi}} &=& 
     (\bra{\psi}_R \otimes \bra{\psi}_{AE} ) ( U_{\overline{E}}^\dagger \otimes I_E ) (I_{\overline{E}} \otimes X_E)  ( U_{\overline{E}} \otimes I_E ) (\ket{\psi}_R \otimes \ket{\psi}_{AE} ) \\ 
     &=& 
     (\bra{\psi}_R \otimes \bra{\psi}_{AE} ) ( U_{\overline{E}}^\dagger U_{\overline{E}} \otimes X_E )  (\ket{\psi}_R \otimes \ket{\psi}_{AE} ) \\ 
     &=& \big( \langle \psi |  \psi \rangle_R \big) \, \big(  \underbrace{\bra{\psi}_{AE} \, (I_A \otimes X_E) \,  \ket{\psi}_{AE}}_{c(X)} \big) ,  
\end{eqnarray*}
where we have used the isometric relation $U_{\overline{E}}^\dagger U_{\overline{E}}= I_{RA} = I_R \otimes I_A$ in the last line above. 
\end{proof} 

\subsection{Non-Natural-Subsystem Replacer Codes}

We next present another, somewhat more interesting example that helps illustrate some of the subtlety involved for codes that cannot be identified with a natural subsystem of the overall Hilbert space of the system. 

\begin{example}\label{nondividingdimensioneg} 
In terms of the theorem notation, this example illustrates the nature of the system $R$ and stresses that it need not always be associated with a particular qudit ``register''. Consider the following encoded basis states for a single qubit code on a two ququart system $\mathcal H = \mathbb{C}^4 \otimes \mathbb{C}^4$ defined as follows:
    \begin{align}
        \ket{\widetilde{0}} = \ket{00}+\ket{11} \notag \\
        \ket{\widetilde{1}} = \ket{20}+\ket{31} \notag
    \end{align}
    
This code is correctable for erasure (recall that means the replacer channel with the completely depolarizing channel) of the second ququart system. This can be done practically using the following process, viewing the two systems as possessed by two parties Alice (1st system) and Bob (2nd system): After the erasure of the second system, Alice applies a projection of her state into the $\{\ket{0},\ket{2}\}$ and $\{\ket{1},\ket{3}\}$ bases. This allows her to recover the information contained in the state. Then it is possible to recover the original state by adding an ancillary system initialized in the $\ket{0}$ state and applying a joint operator. 

Let us view this example from the structure theorem perspective. Here we have $\mathcal H = \overline{E} E = \mathbb{C}^4 \otimes \mathbb{C}^4$ and $R = \mathrm{span}\, \{ \ket{0}_R, \ket{1}_R \}$. We can find $A$ (which can only be $\mathbb{C}$ or $\mathbb{C}^2$) and $\Gamma_A$ as we have done previously by computing the partial trace of code basis states. As a comparison, note that we get different states here: 
\[
\tr_2 ( \kb{\widetilde{0}}{\widetilde{0}} ) = \frac12 (\kb{0}{0} + \kb{1}{1}    )  \quad \mathrm{and} \quad 
  \tr_2 ( \kb{\widetilde{1}}{\widetilde{1}} ) = \frac12 (\kb{2}{2} + \kb{3}{3}    )  .  
\]
As they are rank-2 operators, we have $A=\mathbb{C}^2$, and given their eigenvalues, we have $\Gamma_A = \frac12 I_A$. This further yields $ \sigma_{AE} = \Gamma_A \otimes \sigma_E = \frac18 I_{AE}$, as $\sigma_E = \frac14 I_E$, and $\ket{\psi}_{AE} = \frac{1}{\sqrt{2}} (\ket{00}_{AE}  + \ket{11}_{AE})$ is a purification of $\Gamma_A$. 
We can then define the unitary as follows: 
\[
\begin{array}{cccc} 
R & A & U_{\overline{E}}  & \overline{E}  \\ \hline 
\ket{0} & \left\{ \begin{array}{l}  \ket{0} \\ \ket{1} \end{array} \right. & \longmapsto & \left\{ \begin{array}{l}  
\ket{0} \\ 
\ket{1} \end{array} \right. \\ 
\ket{1} & \left\{ \begin{array}{l}  \ket{0} \\ \ket{1} \end{array} \right. & \longmapsto & \left\{ \begin{array}{l}   
\ket{2} \\ 
\ket{3} 
\end{array} \right. 
\end{array} 
\] 
Indeed, one can verify directly from the table and definition of the code basis states that this yields $\ket{\widetilde{i}} = \big( U_{\overline{E}} \otimes I_E \big) \, ( \ket{i}_R \otimes \ket{\psi}_{A E})$ for $i=1,2$. 

Let us observe the information-theoretic condition~$(iv)$ of the theorem explicitly in this case. We have the state $\ket{\phi}$ defined as: 
    \begin{align}
        \ket{\phi} = \frac{1}{\sqrt{2}}(\ket{0}_{\T{Q}}\ket{\Tilde{0}}_{\B{E}E}+\ket{1}_{\T{Q}}\ket{\Tilde{1}}_{\B{E}E}). 
    \end{align}
    We calculate to find $\rho^{\T{Q}} = \tr_{\overline{E}E}(\kb{\phi}{\phi}) = \frac12 I_{\T{Q}}$ and $\rho^{E} = \tr_{\T{Q}\overline{E}}(\kb{\phi}{\phi}) = \frac12 (\kb{0}{0} + \kb{1}{1})_E$, where note the latter operator is a rank-2 projection on the 4-dimensional space $E$. We can now see that the condition is indeed satisfied:  
    \begin{align}
        \rho^{\T{Q}E} &= \tr_{\B{E}} ( \ket{\phi}\bra{\phi} ) \notag \\
        &= \frac{1}{4}(\ket{00}\bra{00}+\ket{01}\bra{01}+\ket{10}\bra{10}+\ket{11}\bra{11})_{\T{Q}E}  \notag \\
        &= \rho^{\T{Q}}\otimes\rho^E . 
    \end{align}
    As a corollary, also observe that the mutual information of these joint systems is equal to $0$: 
    \[
        I_{\T{Q}E} =S_{\T{Q}}+S_E-S_{\T{Q}E} 
        = \ln{2}+\ln{2}-\ln{4} 
        = 0 . 
    \]

Based on the construction of the unitary above, one can see there is an equivalent example that uses two qubits instead of a single $4$-dimensional qubit, in which the subsystem does correspond to the first qubit register.
Designate two single qubit systems by $1_a, 1_b$ and let the encoded basis states for an erasure correctable code be as follows:
    \begin{align}
        \ket{\Tilde{0}} = (\ket{000}+\ket{011})_{1_a1_bB} \notag \\
        \ket{\Tilde{1}} = (\ket{100}+\ket{111})_{1_a1_bB} \notag
    \end{align}
    The erased system $E$ is still the system $B$, while $R=1_a$ so that the erasure may be corrected.
    The unitary will map to $\ket{00},\ket{01},\ket{10},\ket{11}$ (instead of $\ket{0},\ket{1},\ket{2},\ket{3}$).
With respect to the theorem these two cases are effectively equivalent and satisfy the conditions identically, underscoring the generality of the theorem.
\end{example}    

\subsection{Stabilizer Codes and Cleaning Lemma}

We finish by considering the important class of quantum stabilizer codes, and we revisit the well-known `Cleaning Lemma' for such codes in light of the theorem.

The {\it stabilizer formalism} of Gottesman \cite{gottesman1996class,gottesman1997stabilizer} gives a framework to build and characterize codes for Pauli error models. The starting point for an $n$-qubit stabilizer code $S$ is an Abelian subgroup $\mathcal S$ of the Pauli group $\mathcal P_n$ that does not contain $-I$. The stabilizer subspace for $\mathcal S$, which is the code space, is $S = \mathrm{span} \{ \ket{\psi} \, : \, P \ket{\psi} = \ket{\psi} \,\, \forall\, P \in \mathcal S\}$. The code can encode $k$ logical qubits (i.e., it is $2^k$-dimensional) exactly when $\mathcal S$ has $n-k$ independent generators. The normalizer subgroup $\mathcal N(\mathcal S)$ of $\mathcal S$ inside $\mathcal P_n$ coincides with its centralizer $\mathcal Z(\mathcal S)$, as every element of $\mathcal P_n$ either commutes or anti-commutes and $-I \notin \mathcal S$. 
A main result in the stabilizer formalism asserts that a stabilizer code $S$ defined by a group $\mathcal S$ is correctable for a set of Pauli error operators $\{E_i\}$ exactly when all the operator products $E_i^\dagger E_j$ {\it do not} belong to the set $\mathcal N(\mathcal S) \setminus \langle \mathcal S, iI \rangle$. 

In applications of the stabilizer formalism, one often considers stabilizer codes that are correctable for full erasures of subsets of qubits. Indeed, given a stabilizer code, a subset of qubits $E$ is said to be `correctable' if (in our terminology) the code is correctable for the replacer channel $\mathcal E_E$ with $\mathcal D_E(\rho_E) \propto I_E$ (i.e., the erasure channel on $E$).  
Note that every such channel can be implemented with a set of (unnormalized) Kraus operators drawn from the Pauli group with support contained in $E$, where $\mathrm{supp}(P)\subseteq E$ for $P\in \mathcal P_n$ means that $P$ acts trivially (i.e., as the identity operator) on the complementary qubits $\overline{E}$.  

Thus, we can apply the structure theorem to such codes, and jointly make use of Gottesman's characterization of correctable sets of error operators via the normalizer subgroup to obtain information. We exhibit this perspective by giving the details for a seminal example in the formalism, and then, taking motivation from the situation for stabilizer codes specifically, we prove a general result for replacer codes that follows from the structure theorem. 

\begin{example}
The five-qubit code encoding one logical qubit and correcting any two erasures, discovered independently in~\cite{PhysRevLett.77.198,PhysRevA.54.3824}, is unique up to local equivalences. This code is notable as it is both perfect and MDS, which in particular implies it defines a perfect secret sharing scheme. Following \cite{gottesman1996class,gottesman1997stabilizer}, we choose as a basis for our code $S$ the (unnormalized) states:
\begin{align*}
    \ket{\widetilde{0}} = & \quad\:\ket{00000} + \ket{10010} + \ket{01001} + \ket{10100} + \ket{01010} - \ket{11011} - \ket{00110} - \ket{11000} \\
    & -\ket{11101} - \ket{00011} - \ket{11110} - \ket{01111} - \ket{10001} - \ket{01100} - \ket{10111} + \ket{00101} \\
    \ket{\widetilde{1}} = & \quad\:\ket{11111} + \ket{01101} + \ket{10110} + \ket{01011} + \ket{10101} - \ket{00100} - \ket{11001} - \ket{00111} \\
    & -\ket{00010} - \ket{11100} - \ket{00001} - \ket{10000} - \ket{01110} - \ket{10011} - \ket{01000} + \ket{11010} \\
\end{align*}
The code $S$ is a stabilizer code, stabilized by the group 
\begin{equation*}
  \mathcal S =  \left\langle X\otimes Z\otimes Z\otimes X\otimes I, I\otimes X\otimes Z\otimes Z\otimes X, X\otimes I\otimes X\otimes Z\otimes Z, Z\otimes X\otimes I\otimes X\otimes Z\right\rangle,
\end{equation*}
and has logical operators $\widetilde{X}=X\otimes X\otimes X\otimes X\otimes X$ and $\widetilde{Z}=Z\otimes Z\otimes Z\otimes Z\otimes Z$, both of which observe belong to $\mathcal N(\mathcal S) \setminus \langle \mathcal S, iI \rangle$.

If we consider erasure of qubits 4 and 5, for instance, we can obtain $A$ and $\Gamma_A$ as in the examples above. Calculating we find 
$\mathcal U_{\overline{E}} (\kb{0}{0}_R \otimes \Gamma_A) = \tr_{45} (\kb{\widetilde{0}}{\widetilde{0}} )$ has rank-4 and all its eigenvalues are $\frac14$. 
Hence we take $A = \mathbb C^4$ and $\Gamma_A = \frac14 I_4$ here. Then we can define the unitary $U_{\overline{E}}$ via its action on basis states as follows:
\[
\begin{array}{cccc} 
R & A & U_{\overline{E}}  & \overline{E}  \\ \hline 
\ket{0} & \left\{ \begin{array}{l}  \ket{00} \\ \ket{01} \\ \ket{10} \\ \ket{11} \end{array} \right. & \longmapsto & \left\{ \begin{array}{l} +(\ket{000} - \ket{011} + \ket{101} - \ket{110}) \\ +(\ket{001} + \ket{010} - \ket{100} - \ket{111}) \\ -(\ket{001} - \ket{010} - \ket{100} + \ket{111}) \\ -(\ket{000} + \ket{011} + \ket{101} + \ket{110}) \end{array} \right. \\ 
\ket{1} & \left\{ \begin{array}{l}  \ket{00} \\ \ket{01} \\ \ket{10} \\ \ket{11} \end{array} \right. & \longmapsto & \left\{ \begin{array}{l} -(\ket{001} + \ket{010} +\ket{100} + \ket{111}) \\ -(\ket{000} - \ket{011} - \ket{101} + \ket{110}) \\ -(\ket{000} + \ket{011} - \ket{101} - \ket{110}) \\ -(\ket{001} - \ket{010} + \ket{100} - \ket{111}) \end{array} \right. \\  
\end{array} 
\]

One can verify from the table and definition of the code basis states that for $i=0,1$ we have 
\begin{equation*}
    \ket{\widetilde{i}}=\left(U_{\overline{E}}\otimes I_E\right)\left(\ket{i}\otimes\ket{\psi}_{AE}\right),
\end{equation*}
where here $\ket{\psi}_{AE}$ is two copies of the canonical maximally entangled state between the ancilla and erased qubits:
\begin{equation*}
    \ket{\psi}_{AE}=\frac{1}{2}\left(\ket{00}_A\otimes\ket{00}_E+\ket{01}_A\otimes\ket{01}_E+\ket{10}_A\otimes\ket{10}_E+\ket{11}_A\otimes\ket{11}_E\right).
\end{equation*}

As a lead-in to our discussion below on the Cleaning Lemma, let us note in this case we can find logical operators $\widetilde{X}',\widetilde{Z}'$, anti-commuting and belonging to the normalizer subgroup, whose support is contained in $\overline{E}$:
\begin{align*}
    \widetilde{X}' & = Z\otimes X\otimes Z\otimes I\otimes I \\
    \widetilde{Z}' & = Y\otimes Z\otimes Y\otimes I\otimes I
\end{align*}
Additionally, we can pick a set of stabilizer generators such that for each qubit in $E$, there is a unique generator with an $X$ on the qubit and a unique generator with $Z$ on the qubit, with all other generators having $I$ on the qubit. In this case, we have:
\begin{align*}
    Z\otimes Z\otimes X\otimes I\otimes X \\
    Y\otimes Y\otimes Z\otimes I\otimes Z \\
    X\otimes Z\otimes Z\otimes X\otimes I \\
    Z\otimes Y\otimes Y\otimes Z\otimes I
\end{align*}
\end{example}

The Cleaning Lemma \cite{Rai98,BT09,Ash20,Preskill_notes} for stabilizer codes is one of the most useful tools in the theory of such codes. It states that given a stabilizer code that is correctable for a subset of qubits $E$, for any logical operator $L\in \mathcal N(\mathcal{S})$, there exists a stabilizer operator $S\in \mathcal{S}$ such that $LS$ acts trivially on the qubits of $E$; i.e., $\mathrm{supp}(LS)\subseteq \overline{E}$. 

The following result can be viewed as a generalization of the Cleaning Lemma that applies to all replacer codes (we comment after the proof on this point). By the {\it commutant} $\mathcal S^\prime$ of a set of operators $\mathcal S$, we mean the set of all operators that commute with all the operators of $\mathcal S$; that is, $\mathcal S^\prime = \{ P \, : \, PQ = QP \,\,\, \forall Q\in\mathcal S\}$. We also recall the notation $\mathbb{T} = \{\lambda\in\mathbb{C} \, : \, |\lambda|=1\}$  for the unit circle in the complex plane. 

\begin{cor}\label{cleaninglemma}
Suppose $S$ is a subspace of $\mathcal H = \overline{E} E$ and let $P_S$ be the projection of $\mathcal H$ onto $S$. If $S$ is a correctable code for a replacer channel $\mathcal E_E$, then there are no unitary operators $L\in \{ P_S \}^\prime$ for which $L P_S \notin \mathbb{T} P_S$ and $\mathrm{supp}(L)\subseteq E$.
\end{cor}

\begin{proof}
Suppose we have a unitary operator $L\in \{ P_S \}^\prime$ such that $\mathrm{supp}(L)\subseteq E$, so that $L = I_{\overline{E}} \otimes L_E$ for some unitary $L_E$ on $E$. 

We can use the code state form given in condition~$(iv)$ of the structure theorem to find, for each basis state $ \ket{\widetilde{i}}$ of $S$, 
\begin{eqnarray*}
 L  \, \ket{\widetilde{i}} &=& \big( I_{\overline{E}} \otimes L_E \big) \big( U_{\overline{E}} \otimes I_E \big) \, ( \ket{i}_R \otimes \ket{\psi}_{A E}) \\ 
 &=&  \big( U_{\overline{E}} \otimes I_E \big) \big( I_{RA} \otimes L_E \big) \, ( \ket{i}_R \otimes \ket{\psi}_{A E}) \\ 
  &=&  \big( U_{\overline{E}} \otimes I_E \big) \, ( \ket{i}_R \otimes \big( ( I_{A} \otimes L_E   ) \ket{\psi}_{A E}\big) ) . 
\end{eqnarray*}
On the other hand, since $L$ commutes with $P_S$, there are scalars $\{a_j\}_j$ such that 
\[
L  \, \ket{\widetilde{i}} = L P_S \ket{\widetilde{i}} = P_S L   \ket{\widetilde{i}} = \sum_j a_j  \ket{\widetilde{j}} = \big( U_{\overline{E}} \otimes I_E \big) \, \big( \big( \sum_j a_j \ket{j}_R \big) \otimes \ket{\psi}_{A E} \big).  
\]

Hence we can multiply both descriptions of $L\ket{\widetilde{i}}$ by $U_{\overline{E}}^\dagger \otimes I_E$ and use the fact that $U_{\overline{E}}^\dagger U_{\overline{E}} = I_{RA}$ to obtain for all $ \ket{\widetilde{i}}$, 
\[
\ket{i}_R \otimes \big( ( I_{A} \otimes L_E   ) \ket{\psi}_{A E}\big) = 
\Big( \sum_j a_j \ket{j}_R \Big) \otimes \ket{\psi}_{A E} . 
\]
If we take the inner product of this vector with $\bra{j}_R\bra{\psi}_{AE}$ for any $j$, we get 
\[
\delta_{ij} \big( \,   \underbrace{\bra{\psi}_{AE} (I_{A} \otimes L_E   ) \ket{\psi}_{A E}}_{\lambda} \big) = a_j, 
\]
and note that the indicated scalar $\lambda$ is independent of $i$. (Observe this is the scalar of Corollary~\ref{correctablecor}.) 

Therefore,  we have $L \ket{\widetilde{i}} = \lambda \ket{\widetilde{i}}$ for all $i$, and so $LP_S = P_S L = \lambda P_S$ (also $\lambda\in\mathbb{T}$ as $L$ is unitary). The result follows.  
\end{proof}

\begin{remark}
We can obtain the explicit version of the Cleaning Lemma noted above by applying this result to stabilizer erasure codes. If $S$ is a stabilizer code with stabilizer group $\mathcal S$, then the set of Pauli operators that commute with $P_S$ is exactly the centralizer, $\mathcal Z(\mathcal S) = \mathcal N(\mathcal S)$, of $\mathcal S$ inside $\mathcal P_n$ (this follows from the explicit form for $P_S$ in terms of the elements of $\mathcal S$). Further, the elements of $\mathcal P_n$ that belong to $\mathbb{T} P_S$ when restricted to $S$ in that case are precisely the elements of the group $\langle \mathcal S, iI \rangle$. Hence, Corollary~\ref{cleaninglemma} applied to the stabilizer case says, when the qubits $E$ are correctable for the stabilizer code $S$, the set of elements in $\mathcal Z(\mathcal S)$ supported on $E$ is equal to the set of elements in $\mathcal S$ that are supported on $E$ (in general the former contains the latter, but they coincide in the correctable code case). One can then take the centralizer of both sets and intersect with the operators supported on $E$ to arrive at the explicit form.      

Of course, Corollary~\ref{cleaninglemma} applies to arbitrary replacer codes, and so it is natural to ask if a more explicit form of the generalized Cleaning Lemma holds, such as in the stabilizer case. There are some issues to sort out if one were to attempt to prove this is the case. For instance, while it seems natural that the commutant of the code space projection replaces the normalizer group in the general case, it is not immediately clear what operator set should replace the stabilizer group. We leave this as an open problem. 
\end{remark}

\section{Conclusion}\label{sec:conclusion}

Our main result can be viewed as a structure theorem for quantum replacer codes in that it gives multiple equivalent descriptions of such codes from different perspectives, each of which is useful in its own right. We discussed the conceptual viewpoint of the result and showed how to practically compute the key elements that define these codes. We then presented several examples and applications of the theorem, which we think give a glimpse of the potential utility of the result. We see this work as opening up a number of potential new lines of investigation. In addition to the topics of erasure conversion \cite{WKST22,Ketal23,KCB23} and black hole theory \cite{AXH15} mentioned above, we briefly discuss some other potential directions below. 

While we proved the theorem for replacer errors acting on a fixed subset of qudits, the theorem can be applied equally to a code that is correctable for replacer errors on multiple subsets of qudits. Most prominent amongst such scenarios are codes associated with quantum secret sharing schemes \cite{cleve1999share}, and indeed, some of our examples are of this type. 
We fully expect the theorem can be pushed further in such situations; for instance, presumably one would find constraints imposed on the possible ancilla $A$ and states $\ket{\psi_{AE}}$ and $\Gamma_A$ by requiring the code to be correctable across multiple sets of qudits. Exactly what this might mean for secret sharing schemes remains to be explored.    

In certain physically-relevant settings, considering only single erasures can still be interesting.
Ongoing research into different platforms for building quantum computers and networks connecting them has yielded much progress, but no one contender stands out as sufficient for all
desired outcomes.
Various implementation candidates for qubits include photons (for example, photonic chips acting as processors, and connections built by optical fibers) and solid-state systems such as superconductors, trapped ions, and neutral atoms.
Each has its own strengths and weaknesses (e.g., long/short coherence times, high/low loss, etc), and so researchers have started to combine these platforms, harnessing the strengths of each and using them together to counter the disadvantages \cite{Awshalometal21,Elshaarietal20,Chiaetal24}.
It is conceivable that consideration of only specific erasures may be necessary in this context.
For example, an atomic qubit with a long coherence time, whose presence in an array can be verified, can more confidently be assumed to not be subject to an erasure error.
On the other hand, we might desire for the atomic qubit to interact
with photonic qubits for the purpose of networking over longer distances; however, these photons are subject to a high degree of loss. Subsequently, these qubits could correspond to those with
potential for erasure, which must be corrected.

We included some corollaries of the theorem that give alternative derivations of established results, with new details provided by the theorem in some cases. One could push this idea further, and revisit other results (e.g., no-go results for erasure codes \cite{grassl1997codes}) in light of the theorem and explore for new approaches and details in the results. We also expect the general perspective of the theorem will lead to generalizations of certain results for erasure codes. Our attempt to extend the Cleaning Lemma \cite{Rai98,BT09,Ash20,Preskill_notes} beyond stabilizer codes is an indication of this possibility.  
There are other physically-relevant potential extensions of this work.  
For instance, considering only particular erasures may give rise to a description of biased error correcting codes, and one could imagine an extension of our theorem to that setting. Another real-world consideration that would suggest an extension of the theorem would be in the scenario of limited  reconstruction capabilities of the information, with a physical restriction to local operations and classical communication (LOCC).

Further, we have not considered quantum subsystem codes \cite{kribs2005unified,poulin2005stabilizer,nemec2023quantum} or hybrid classical-quantum codes \cite{nemec2021infinite,dauphinais2024stabilizer} in this paper, but we expect the structure theorem can be generalized for such codes (for instance, the notion of unitarily recoverable extends to such codes \cite{KS06,beny2008algebraic}). This suggests that it should be possible to formulate an appropriate notion of quantum subsystem secret sharing schemes, for instance, and then explore them for potential advantages as in the case for quantum error correcting codes. 

We plan to pursue some of these investigations elsewhere, and we invite others to do the same.

\vspace{0.1in} 

{\noindent}{\bf Acknowledgments.} D.W.K. was partially supported by NSERC Discovery Grant RGPIN-2024-400160.  E.C and S.H. were supported by NSF grant \#  2112890.  M.N. was supported by NSF grant \# 2016136.  A.N. would like to thank Mingyu Kang for fruitful discussions.

\bibliography{refs}
\bibliographystyle{unsrt}

\end{document}